\documentclass[10pt,final,journal,twocolumn]{IEEEtran}
\usepackage{amsmath}
\usepackage{subfigure}
\usepackage{graphicx}
\usepackage{latexsym}
\usepackage{amssymb,amsbsy,verbatim,array}
\usepackage{epstopdf}
\usepackage{cite}
\usepackage{color}
\usepackage{cleveref}
\newtheorem{remark}{Remark}
\newtheorem{theorem}{Theorem}

\newtheorem{proof}{Proof}

\everymath{\displaystyle}\begin{document}
\title{On Performance of Fluid Antenna System using Maximum Ratio Combining}

\author{Xiazhi Lai, Tuo Wu, Junteng Yao, Cunhua Pan, \emph{Senior Member, IEEE},\\
Maged Elkashlan, \emph{Senior Member, IEEE}, and Kai-Kit Wong, \emph{Fellow, IEEE}
\vspace{-5mm}

\thanks{\emph{(Corresponding author: Tuo Wu and Cunhua Pan.)}}
\thanks{X. Lai is with the School of Computer Science, Guangdong University of Education, Guangzhou, Guangdong, China (E-mail: xzlai@outlook.com).}
\thanks{T. Wu and M. Elkashlan are with the School of Electronic Engineering and Computer Science at Queen
Mary University of London, London E1 4NS, U.K. (Email:\{tuo.wu, maged.elkashlan\}@qmul.ac.uk).}
\thanks{J. Yao is with the Faculty of Electrical Engineering and Computer Science, Ningbo University, Ningbo 315211, China (E-mail: juntengyao512@163.com).}
\thanks{C. Pan is with the National Mobile Communications Research Laboratory, Southeast University, Nanjing 210096, China. (e-mail: cpan@seu.edu.cn).}
\thanks{K. K. Wong is with the Department of Electronic and Electrical Engineering, University College London, WC1E 6BT London, U.K., and also with the Yonsei Frontier Laboratory and the School of Integrated Technology, Yonsei University, Seoul 03722, South Korea (e-mail: kat-kit.wong@ucl.ac.uk).}

}

\markboth{}
{Lai \MakeLowercase{\textit{et al.}}: }

\maketitle

\thispagestyle{empty}

\begin{abstract}
This letter investigates a fluid antenna system (FAS) where multiple ports can be activated for signal combining for enhanced receiver performance. Given $M$ ports at the FAS, the best $K$ ports out of the $M$ available ports are selected before maximum ratio combining (MRC) is used to combine the received signals from the selected ports. The aim of this letter is to study the achievable performance of FAS when more than one ports can be activated. We do so by analyzing the outage probability of this setup in Rayleigh fading channels through the utilization of Gauss-Chebyshev integration, lower bound estimation, and high signal-to-noise ratio (SNR) asymptotic approximations. Our analytical results demonstrate that FAS can harness rich spatial diversity, which is confirmed by computer simulations.
\end{abstract}

\begin{IEEEkeywords}
Diversity, fluid antenna system (FAS), maximum ratio combining (MRC), outage probability.
\end{IEEEkeywords}

\section{Introduction}
Fluid antenna system (FAS) capitalizes upon the inherent spatial diversity by dynamically adjusting the antenna elements to optimal positions, referred to as ``ports". This new paradigm stands in contrast to traditional communication methodologies, in which the antenna elements remain in fixed positions, as elucidated by Shojaeifard {\em et al.}~in \cite{Shojaeifard}. The realization of FAS may come in the forms of liquid-metal-based antennas \cite{Huang21} or on-off pixel-based antennas \cite{Rodrigo14}. See \cite{Wong-frontiers22} for more details.

Motivated by the great potential of FAS, recent research has delved into the FAS channel model, deriving the probability density function (PDF) of the received signal-to-noise ratio (SNR) as well as the corresponding outage probability \cite{FAS20,FAS21,FAS22}. Remarkably, the outcomes of their investigation unveiled the superiority of the FAS scheme over conventional fixed-position antenna systems, particularly when a considerable multitude of ports is at disposal. Machine learning techniques have also been shown to be effective in port selection for FAS \cite{Chai22}. Most recently, Wong \emph{et al.}~has extended the use of FAS for multiple access by taking advantage of the ups and downs of fading channels in the spatial domain, and illustrated the possibility of alternative multiple access schemes using FAS \cite{FAMS,FAMS23,Waqar23}.

However, research in FAS is still in an early stage and the majority of the results so far are limited to FAS with only one selected port exhibiting the maximum SNR \cite{FAS20,FAS21,FAS22,Chai22,FAMS,FAMS23}. The fact that a mobile terminal can actually afford more than one radio frequency (RF) chains, means that it is increasingly probable that FAS can come with multiple activated ports, with better performance \cite{MFAS23}. Since maximum ratio combining  (MRC) is the optimal mixing scheme without interference, it is therefore of great importance to understand the achievable performance of FAS using MRC if more than one ports can be selected for reception. This is the aim of this letter.

Specifically, our contributions are summarized as follows:
\begin{itemize}
\item First, we consider a $K$-port FAS which corresponds to a FAS with $K$ selected ports, operating in Rayleigh fading channels. The mobile receiver selectively activates $K$ optimal ports from the available $M$ ports. Then MRC is employed to combine the $K$ branches of signals from the activated ports. We derive the outage probability of the proposed $K$-port FAS using both Laplace transform (LT) and Gauss-Chebyshev integration methods.
\item Additionally, we present the lower bound and asymptotic expressions for the outage probability.
\item The simulation results substantiate the effectiveness of the proposed analytical approach, thereby confirming and validating our insights and discussions.
\end{itemize}

\section{System Model}
Consider an end-to-end communication in Rayleigh fading channels, where the source transmits the signal using a conventional fixed-position antenna with transmit power $P_S$ but the receiver is equipped with a FAS with $K$ fluid antenna elements.\footnote{In our idealized mathematical model, a FAS with multiple single-activated-port fluid antennas is equivalent to a FAS with multiple activated ports although their specific implementation details will differ.} Each antenna element is connected to one RF chain. Within this particular FAS configuration, a linear space of $W\lambda$ encompasses a total of $M$ ports, where $\lambda$ represents the wavelength \cite{FAS20}. Among these $M$ ports, it is assumed that each port is evenly distributed, and $K$ ports can be activated for signal receiving out of the total $M$ ports.

Since each port is placed closely, the channel parameters of each port are correlated. Building upon the channel model developed in \cite{FAS22} and \cite{FAMS23},  we introduce a virtual reference port to model the channel correlation. This virtual reference port is characterized by a channel parameter $h_0\sim\mathcal{CN}(0,\alpha)$, following a complex Gaussian distribution with zero mean and variance $\alpha$. Accordingly, the SNR of $h_0$ can be written as
 \begin{align}\label{ee1}
\gamma_{0}=\frac{P_S|h_{0}|^2}{\sigma^2},
\end{align}
where $\sigma^2$ denotes the noise power level. Considering $h_0$ as a complex Gaussian random variable (RV),  the PDF of $\gamma_{0}$ can be expressed as
\begin{align}\label{e7}
f_{\gamma_{0}}(x)=\frac{1}{\phi}e^{-\frac{x}{\phi}},
\end{align}
where $\phi=P_S\alpha/\sigma^2$ represents the average received SNR. Now, we proceed to establish the channel parameter linking the source and the $m$-th port, denoted as $h_{m}$, where $m\in\mathcal{M}=\{1,2,\dots,M\}$. The expression for $h_{m}$ takes the form
\begin{align}\label{e1}
h_{m}=\mu h_{0}+(1-\mu) e_{m},
\end{align}
where  $e_{m}\sim\mathcal{CN}(0,\alpha)$ for $m\in\mathcal{M}$ are independently and identically distributed (i.i.d.) RVs, $\alpha$ is the average channel gain from the source to the ports. Additionally, $\mu$ denotes the correlation factor, which is given by \cite{FAS22}
\begin{align}\label{e2}
\mu=&\sqrt{2}\sqrt{{}_1F_2\Big(\frac{1}{2};1;\frac{3}{2};-\pi^2W^2\Big)-\frac{J_1(2\pi W)}{2\pi W}},
\end{align}
where ${}_a F_b$ denotes the generalized hypergeometric function and $J_1(\cdot)$ is the first-order Bessel function of the first kind.

Conditioned on a fixed channel parameter $h_{0}$, and in accordance with $\gamma_{0}$, the corresponding SNR of $h_{m}$, expressed as $\gamma_{m}=\frac{P_S|h_m|^2}{\sigma^2}$, follows a non-central chi-square distribution. The conditional PDF can be expressed as
\begin{align}\label{e3}
f_{\gamma_{m} | \gamma_{0}=x_0}(x)=&\omega
e^{-\omega(x+\mu x_{0})}
I_{0}\big(2\omega\sqrt{\mu x_{0} x}\big),
\end{align}
where  $\omega=\big(\phi(1-\mu)\big)^{-1}$. Besides, $I_{0}(u)$ is the modified Bessel function of the first kind with order $0$, which  can be expressed in series representation as \cite{book}
\begin{align}\label{e4}
I_0(z)=\sum_{k=0}^{\infty}\frac{z^{2k}}{2^{2k} k!\Gamma(k+1)}.
\end{align}
Combining \eqref{e3} with \eqref{e4}, we further derive $f_{\gamma_{m} | \gamma_{0}=x_0}(x)$ as
\begin{align}\label{e5}
f_{\gamma_{m} | \gamma_{0}=x_0}(x)=&
\sum_{k=0}^{\infty}c_k x_0^k e^{-\omega\mu x_{0}} x^k e^{-\omega x},
\end{align}
where
\begin{align}\label{e6}
c_k=\frac{\omega^{2k+1}\mu^{k}}{(k!)^2}.
\end{align}

In order to receive the signal transmitted from the source, the receiver selects the $K$ ports with the $K$ highest received SNR from the available total of $M$ ports for activation.  The set of selected ports is denoted by
\begin{align}
\mathbb{K}=\arg \emph{\mbox{K}}\max_{m\in \mathcal{M}}\gamma_m,
\end{align}
where $\emph{\mbox{K}}\max_{m\in \mathcal{M}} \gamma_m$ denotes to select the $K$ maximal $\gamma_m$ out of set $\mathcal{M}$. In addition, to process the received signals from different antenna elements, the MRC technique is utilized to combine the $K$ branches of signals.

Moreover, the channel state information (CSI) is assumed to be not available at the source; hence the transmission data rate is fixed to $R$. Therefore, the outage of communication occurs when the FAS cannot sustain the data rate $R$, i.e.,
\begin{align}\label{e8}
\mathcal{E}=\left\{ \log_2\left(1+\sum_{m\in \mathbb{K}}\gamma_m\right)\leq R\right\}.
\end{align}
Thus, the system's outage probability is written as
\begin{align}\label{e9}
P_{\rm{out}}=\Pr\left(\mathcal{E}\right).
\end{align}

\section{Performance Analysis}
Here, we derive the exact outage probability of the proposed FAS-enabled communications. Subsequently, the lower bound and asymptotic expressions of the outage provability of system are derived. These derivations offer valuable insights for the proposed FAS-enabled communications system.

\subsection{Exact Outage Probability}
Consider the port with the $(K+1)$-th maximal channel gain, denoted as $v$. Given $\gamma_0=x_0$, the outage probability is expressed as
\begin{align}\label{q1}
\Lambda(z)&\hspace{.5mm}=\Pr\left(\sum_{m\in\mathbb{K}}\gamma_{m}\leq z | \gamma_{0}=x_0 \right)\notag\\
&\overset{(a)}{=}\binom{M}{K}(T+1)\int_{0}^{\infty}\Phi(z)\Psi(v)f_{v|\gamma_0=x_0}(v)dv,
\end{align}
where $z=2^R-1$ denotes the SNR threshold of outage, $T=M-K-1$, and
\begin{align}\label{q2}
\Psi(v,x_0)=\Pr\left(\gamma_{m}\leq v, m\in\mathcal{T}| \gamma_0=x_0 \right),
\end{align}
is the probability that $T+1$ ports are idle with maximal channel gain $v$, and $\mathcal{T}=\{1,2,\dots, T\}$. Also,
\begin{align}\label{q3}
\Phi(z,v,x_0)=\Pr\left(\sum_{m\in\mathcal{K}}\gamma_{m}\leq z,\gamma_m>v| \gamma_{0}=x_0 \right),
\end{align}
is the probability that $K$ ports are selected and outage occurs and $\mathcal{K}=\{1,2,\dots, K\}$. Step ($a$) holds since $\gamma_m$ for $m\in\mathcal{M}$ are i.i.d.~RVs, and $\Psi(v,x_0)\Phi(z,v,x_0)$ represents  the outage probability related to one of the port selection results.

In the following, we derive the expressions of $\Psi(v,x_0)$ and $\Phi(z,v,x_0)$. Then we obtain the outage probability by taking the expectation of $\Lambda(z)$ with respect to $\gamma_0$.

First, it is important to note that $\forall m,l \in\cal{T}$, $\gamma_{m}$ and $\gamma_{l}$ are independent with each other given $\gamma_{0}=x_0$. Furthermore, in accordance with \eqref{e3}, the joint PDF of $\gamma_{m}$ for $m\in\cal{T}$ can be expressed as
\begin{multline}\label{q4}
f_{\gamma_{m}, m\in\mathcal{T }| \gamma_{0}=x_0}(x_1,\dots, x_T)\\
=\prod_{m=1}^{T}\omega
e^{-\omega(x_m+\mu x_0)}
I_{0}\big(2\omega\sqrt{\mu x_0 x_m}\big).
\end{multline}
Then, by utilizing \eqref{q2} and \eqref{q4}, we evaluate  $\Psi(v,x_0)$ as
\begin{align}\label{q5}
\Psi(v,x_0)&=\int^v_{0} \cdots \int^v_{0} f_{\gamma_{m}, m\in\mathcal{T }| \gamma_{0}=x_0}( x_1,\dots, x_T)dx_1\cdots dx_T \notag\\
&=\Big(1-Q_1\big(\sqrt{2\omega\mu x_0},\sqrt{2\omega v}\big)\Big)^T,
\end{align}
where $Q_1(\cdot,\cdot)$ is the first order Marcum-$Q$ function \cite{FAS21}.

Next, we proceed to derive the analytical expression of $\Phi(z,v,x_0)$ by utilizing the following theorem.

\begin{theorem}
The LT expressions of the following functions
\begin{align}\label{q61}
g(x)&=x^a e^{-bx}u(x-v),\\\label{q62}
p(x)&=(x-a)^{K-1}e^{-bx}u(x-a),
\end{align}
are, respectively,
\begin{align}\label{q7}
L[g(x);s]&=e^{-(s+b)v}\sum_{l=0}^{a}\frac{a!v^l}{l!(s+b)^{a+1-l}},\\
L[p(x);s]&=\frac{(K-1)!e^{-a(s+b)}}{(s+b)^K},
\end{align}
where $\mathrm{Re}(s) \geq -b$, $\rm{Re}(x)$ denotes the real part of $x$, and $u(\cdot)$ is the step function.
\end{theorem}

\begin{proof}
See Appendix A.
\end{proof}

From Theorem 1 and \eqref{e5}, the LT of the PDF of  $ \gamma_{m}$ with $ \gamma_{m}>v$ is given by
\begin{multline}\label{q8}
L\big[f_{\gamma_m|\gamma_0=x_0}(x_m);s\big]\\
= e^{-(s+\omega)v-\omega\mu x_0}\sum_{m=0}^{\infty}\sum_{l=0}^{m}\frac{d_m x_0^m v^l}{l!(s+\omega)^{m+1-l}},
\end{multline}
where $\rm{Re}(s)\geq -\omega$ and $d_m=c_m m!$.

Then, by using the faltung theorem in \cite{book},
the LT of the PDF of RV $\bar{\gamma}=\sum_{m=1}^K\gamma_{m}$ conditioned on $\gamma_m>v$ can be derived as
\begin{align}\label{q10}
&L\big[f_{\bar{\gamma}|\gamma_0=x_0}(x);s\big]\notag\\
&=\Big(L\big[f_{\gamma_{m}|\gamma_0=x_0}(x_m);s\big]\Big)^K    \notag\\
&=e^{-Kv(s+\omega)-K\omega\mu x_0}
\sum_{r_m=0\atop m\in\mathcal{K}}^{\infty}\rho_m x_0^{\eta_m}
\sum_{l_m=0\atop m\in\mathcal{K}}^{r_m}\frac{v^{\epsilon_m}q_m}{(s+\omega)^{\chi_m}},
\end{align}
where
\begin{equation}\label{q12}
\left\{\begin{aligned}
\rho_m&=\prod_{m=1}^{K}d_m,\\
\eta_m&=\sum_{m=1}^{K}r_m,\\
\epsilon_m&=\sum_{m=1}^K l_m,\\
q_m=&\prod_{m=1}^K\frac{1}{l_m!},\\
\chi_m&=K+\eta_m-\epsilon_m.
\end{aligned}\right.
\end{equation}

Utilizing Theorem 1, we can obtain the PDF of $\bar{\gamma}$ conditioned on $\gamma_0=x_0$ as
\begin{multline}\label{q13}
f_{\bar{\gamma}|\gamma_0=x_0}(x)=e^{-\omega(x+K\mu x_0)}
\sum_{r_m=0\atop m\in\mathcal{K}}^{\infty}\rho_m x_0^{\eta_m}\\
\times\sum_{l_m=0\atop m\in\mathcal{K}}^{r_m}v^{\epsilon_m}q_m\frac{\big(x-Kv\big)^{\chi_m-1}}{(\chi_m-1)!},
\end{multline}
with $x\geq Kv$. Based on \eqref{q13}, the computation of $\Phi(z,v,x_0)$ can be performed by
\begin{align}\label{q14}
\Phi(z,v,x_0)&=\int_{Kv}^{z}    f_{\bar{\gamma}|\gamma_0=x_0}(x)dx\notag\\
&=e^{-\omega(Kv+K\mu x_0)}
\sum_{r_m=0\atop m\in\mathcal{K}}^{\infty}\rho_m x_0^{\eta_m}\notag\\
&\quad\times\sum_{l_m=0\atop m\in\mathcal{K}}^{r_m}v^{\epsilon_m}q_m
\frac{\gamma\big(\chi_m,\omega(z-Kv)\big)}{(\chi_m-1)!\omega^{\chi_m}},
\end{align}
in which $z\geq Kv$ is a necessary condition; otherwise, $\Phi(z,v,x_0)=0$. In addition, $\gamma(\alpha,x)$ is the lower incomplete Gamma function, which can be expressed in integral and serial representations respectively, as
\begin{align}\label{q15}
\gamma(\kappa,x)=\int_0^{x}e^{-t}t^{\kappa-1}dt=(\kappa-1)!\left(1-e^{-x}\sum_{m=0}^{\kappa-1}\frac{x^m}{m!}\right).
\end{align}
Calculating $\Lambda(z)$ in \eqref{q1} with \eqref{q5} and \eqref{q14}, and then taking the expectation of $\Lambda(z)$ with respect to $\gamma_0$, the outage probability of the system can be computed as
\begin{multline}\label{q16}
P_{\mathrm{out}}=\int_{0}^{\infty}\int_{0}^{\frac{z}{K}}\binom{M}{K}(T+1)\Phi(z,v,x_0)\Psi(v,x_0)\\
\times f_{\gamma_m|\gamma_0=x_0}(v)f_{\gamma_0}(x_0)dvdx_0.
\end{multline}

\begin{remark}
From \eqref{q5}, it becomes evident that $\Psi(v,x_0)$ becomes tiny with a large number of $T$, owing to the fact that $Q_1(\cdot,\cdot)$ is bounded by 1 \cite{FAS21}.  This observation implies that $P_{\mathrm{out}}$ in \eqref{q16}, i.e., the outage probability of the system approaches zero when the total number of ports $M\rightarrow\infty$.\footnote{Note that the conclusion may vary depending on how spatial correlation over the ports is modelled. That said, the analysis presented in this letter gives the first-look performance of FAS using MRC.}
\end{remark}

It is noticeable that the integral in \eqref{q16} presents computational challenges. To address this, we initially replace the upper limit of the integral in \eqref{q16} with a sufficiently large value denoted as $H$.  This approximation is valid because the integrand in \eqref{q16} tends to approach zero as $x_0$ increases. Subsequently, we resort to  the Gauss-Chebyshev integral to derive a precise approximation of $P_{\mathrm{out}}$ in serial representation:
\begin{multline}\label{q17}
P_{\mathrm{out}}
\approx \binom{M}{K}\frac{\pi^2 H z(T+1)}{4U_pU_l}
\sum_{p=1}^{U_p}
\sum_{l=1}^{U_l}
\Phi(z,y_l,y_p)\Psi(y_l,y_p)\\
\times \sqrt{1-t_p^2}\sqrt{1-t_l^2}
f_{\gamma_m|\gamma_0=y_p}(y_l)f_{\gamma_0}(y_p),
\end{multline}
where $U_p$ and $U_l$ are complexity-accuracy tradeoff parameters, and
\begin{equation}\label{q18}
\left\{\begin{aligned}
t_p&=\cos\left(\frac{(2p-1)\pi}{2U_p}\right),\\
y_p&=\frac{H (t_p+1)}{2},\\
t_l&=\cos\left(\frac{(2l-1)\pi}{2U_l}\right),\\
y_l&=\frac{z (t_l+1)}{2K}.
\end{aligned}\right.
\end{equation}
According to \cite{NumericalAnalysis}, it is established that the approximation provided in \eqref{q17} is tight with large numbers of $U_p$ and $U_l$.

\subsection{Lower Bound and Asymptotic Analysis}
For the sake of facilitating computation and analysis of $P_{\rm{out}}$,  we derive a lower bound for $P_{\rm{out}}$ in this subsection. Notably, this lower bound closely approximates the exact outage probability, particularly in the high SNR region. Moreover, we analyze the asymptotic behavior of $P_{\rm{out}}$ and discuss the performance bottleneck of the system.

First, from \eqref{e5}, we can readily know that $f_{\gamma_{m} | \gamma_{0}=x_0}(x)$ is lower-bounded by
\begin{align}\label{b1}
\bar{f}_{\gamma_{m} | \gamma_{0}=x_0}(x)=
\omega
e^{-\omega\mu x_{0}}  e^{-\omega x}.
\end{align}
Based on \eqref{b1}, we can accordingly obtain the lower bound of $\Psi(v,x_0)$ and $\Phi(z,v,x_0)$, respectively, as
\begin{align}\label{b2}
\bar{\Psi}(v,x_0)
&=e^{-\omega T\mu x_{0}}
\sum_{t=0}^T \binom{T}{t}(-1)^t  e^{-\omega tv},
\\ \label{b3}
 \bar{\Phi}(z,v,x_0)
&=e^{-\omega(Kv+K\mu x_0)}
-e^{-\omega(z+K\mu x_0)}\notag\\
&\times\sum_{k=0}^{K-1}
\frac{\omega^k}{k!}\sum_{m=0}^k\binom{k}{m}z^{k-m}(-Kv)^{k}.
\end{align}
Applying \eqref{b1}--\eqref{b3} into \eqref{q16}, we can obtain
\begin{align}\label{b4}
P_{\rm{out}}&\geq\bar{P}_{\rm{out}}\notag\\
&=\binom{M}{K}\frac{T+1}{M\mu\omega\phi+1}\notag\\
&\times\left(\sum_{t=0}^T  \binom{T}{t}\beta_t
-
\sum_{t=0}^T\sum_{k=0}^{K-1}
\sum_{m=0}^k
\binom{T}{t}\binom{k}{m} \kappa_{t,k,m}
\right),
\end{align}
where
\begin{align}\label{b5}
\beta_t&=\frac{(-1)^t}{ (t+K+1)}\big(1-e^{-z\omega (t+K+1)}\big),\\
\kappa_{t,k,m}&=\frac{(-1)^{t+m}K^m (z\omega)^{k-m}
\gamma\left(m+1,\frac{z\omega(t+1)}{K}\right)
}{
k!(t+1)^{m+1}
}.
\end{align}
From \eqref{e5}, it becomes apparent that the exact value of $P_{\rm{out}}$ approaches the lower bound $\bar{P}_{\rm{out}}$ as  the average received SNR becomes large, i.e., the values of $\alpha$ or $P_S$ are large. Moreover, when the value of $M$ and $K$ are large, the calculations of outage probabilities in \eqref{q16} and \eqref{q17} become intricate. In contrast, the evaluation of $\bar{P}_{\rm{out}}$ using the expression in \eqref{b4} remains computationally efficient, aiding in the analysis of the  performance of the proposed system.

Furthermore, by applying the expansion $e^{-x}=1-x$ for tiny value of $|x|$, we can obtain the asymptotic expression of outage probability in the high SNR region as
\begin{align}\label{b6}
P_{\rm{out}} \simeq\psi(z\omega)^M,
\end{align}
where
\begin{multline}\label{b7}
\psi=\\
\binom{M}{K}\frac{(T+1)(1-\mu)}{K!(M\mu+1-\mu)K^{T+1}}
\sum_{k=0}^K \binom{K}{k}\frac{1}{k+T+1}.
\end{multline}

\begin{remark}
The asymptotic outage probability in \eqref{b6} indicates that the diversity order of the FAS-aided communication system is $M$. This means that the proposed system can fully exploit the diversity offered by  total available $M$ ports, regardless of the number of activated ports $K$. Therefore, enhancing the system's performance by increasing the number of $K$ ports is feasible; yet the improvement is less significant than the improvement of increasing $M$.
\end{remark}

\section{Numerical Results}
In this section, we present several numerical results for the FAS-aided communications. Following a similar approach to the work in \cite{FAMS}, we assume the value of $W=5$, which is a common choice for 5G networks in the context of handset devices. Moreover, we set the data rate $R$ to $5$ bit/s/Hz, leading to an outage SNR threshold $z$ set at 31. Unless specified otherwise, we refer to the outcomes of our simulations as ``Simul". Also, we denote the results obtained from \eqref{q17}, \eqref{b4}, and \eqref{b6} as ``Ana.", ``LB", and ``Asy.", respectively.

\begin{figure}[htbp]
	\centering
	\includegraphics[width=0.9\linewidth]{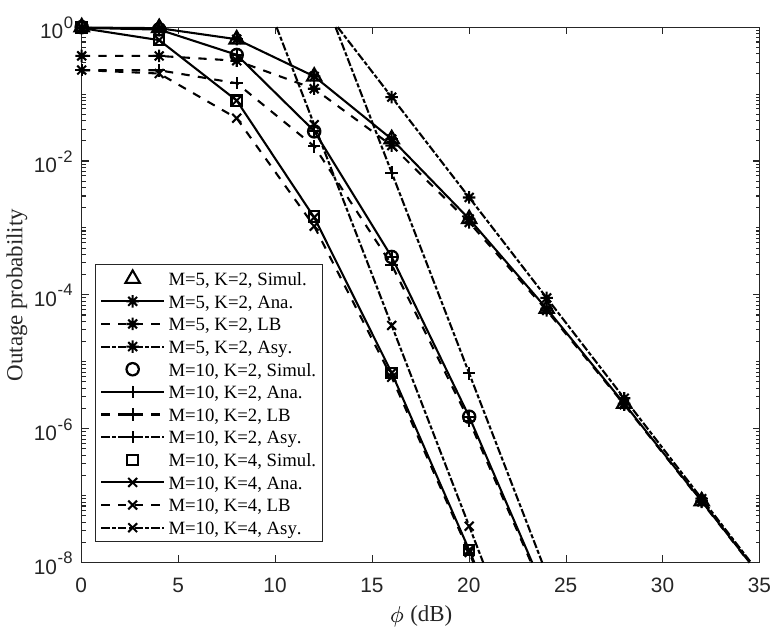}
	\caption{Outage probability versus average SNR $\phi$.}
	\label{fig1}
\end{figure}

Fig.~\ref{fig1}  illustrates the variations in outage probability with the average SNR ($\phi$), considering different values of $M$ and $K$. As observed from Fig.~\ref{fig1}, it is evident that the analytical outage probability derived from equation \eqref{q17} closely aligns with the simulation results.  Additionally, the lower bound provided by equation \eqref{b4} accurately approximates the simulation outcomes, particularly in the high SNR region, which corroborates with the asymptotic result in equation \eqref{b6}.

Furthermore, Fig.~\ref{fig1} indicates that the outage probability of the system is predominantly influenced by the total port number $M$, affirming the analysis in equation \eqref{b6} that the diversity stemming from all available ports can be maximally exploited. It is worth noting that the gain achieved by increasing $K$ from $2$ to $4$ in the high SNR region is approximately $3.8$, consistent with the findings presented in equation \eqref{b6}. However, the enhancement resulting from increasing the number of activated ports $K$ is comparatively less pronounced than the gains derived from increasing $M$.

\begin{figure}[htbp]
	\centering
	\includegraphics[width=0.9\linewidth]{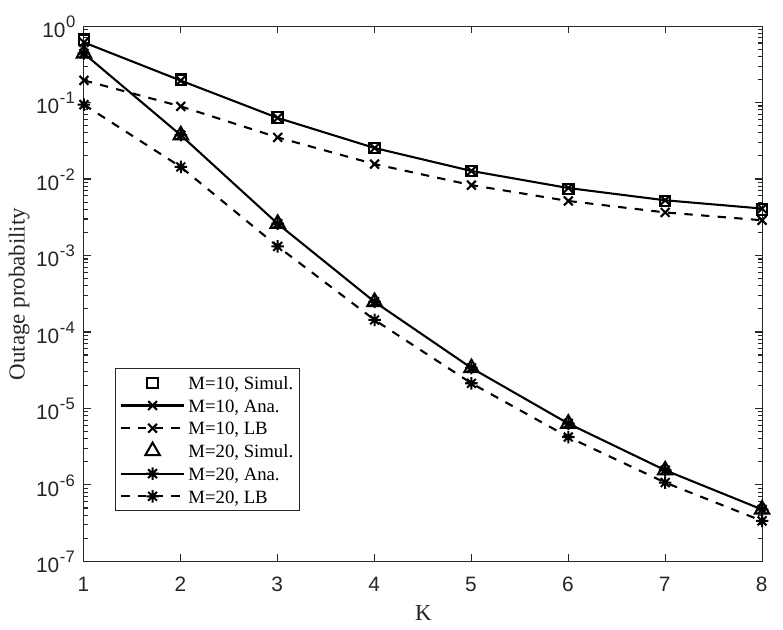}
	\caption{Outage probability versus number of antenna elements $K$.}
	\label{fig2}
\end{figure}

Fig.~\ref{fig2} provides a visualization of the relationship between the number of activated ports ($K$) and the resulting outage probability in the context of the $K$-port FAS-aided communications system. The experiment is conducted with $\phi$ set at $10$ dB, and two distinct values for the total port count ($M$), namely $10$ and $20$. Meanwhile, the number of activated ports $K$ is allowed to vary within the interval of $1$ to $8$.

Upon examining the results depicted in Fig.~\ref{fig2}, it becomes evident that increasing the count of activated ports ($K$) contributes significantly to enhancing the overall system's outage performance. However, the most striking insight emerges from the clear trend indicating that the advantages stemming from augmenting the total port count ($M$) are even more pronounced. This noteworthy pattern is in concordance with the analytical findings presented in the preceding sections.

\section{Conclusion}
In this letter, we proposed to analyze the FAS-aided communications system with multiple activated ports, where the MRC technique was utilized to combine the signal from different activated ports. The outage probability of the proposed system has been derived  in Rayleigh fading channels, in forms of exact expression, lower bound, and asymptotic expression. Analysis showed that the diversity order of the system equals the number of total available ports. Simulation results corroborated the effectiveness of the provided analysis.

\appendices
\section{Proof of Theorem 1}
According  to the definition of LT, we can compute the LT expression of  $g(x)$ in \eqref{q61} as
\begin{align}\label{q7}
L[g(x);s]&=\int_{0}^{\infty}g(x)e^{-sx}dx
=\int_{v}^{\infty}x^a e^{-(s+b)x}dx
\notag\\&
=e^{-(s+b)v}\sum_{l=0}^{a}\frac{a!v^l}{l!(s+b)^{a+1-l}},
\end{align}
where $\mathrm{Re}(s) \geq -b$, and the last step can be derived by using the partial integral technique.

Similarly, we can compute the LT expression of $p(x)$ as
\begin{align}\label{q7}
L[p(x);s]&\hspace{1mm}=\int_{0}^{\infty}p(x)e^{-sx}dx\notag\\
&\hspace{1mm}=\int_a^{\infty}(x-a)^{K-1}e^{-(b+s)x}dx\notag\\
&\overset{(e_1)}{=}\frac{e^{-a(s+b)}}{(s+b)^K}\int_0^{\infty}t^{K-1}e^{-t}dt\notag\\
&\overset{(e_2)}{=}\frac{(K-1)!e^{-a(s+b)}}{(s+b)^K},
\end{align}
where $\mathrm{Re}(s) \geq -b$, the step ($e_1$) can be obtained by setting $t=(x-a)(s+b)$, and step ($e_2$) uses the partial integral technique.

\end{document}